%% file: main.tex
\documentclass[sigconf]{acmart}
\AtBeginDocument{%
  }

\setcopyright{acmlicensed}
\copyrightyear{2018}
\acmYear{2018}
\acmDOI{XXXXXXX.XXXXXXX}
\acmConference[Conference acronym 'XX]{Make sure to enter the correct
  conference title from your rights confirmation email}{June 03--05,
  2018}{Woodstock, NY}
\acmISBN{978-1-4503-XXXX-X/2018/06}

\input{math_commands.tex}

\usepackage{algorithm}
\usepackage{algorithmic}
\usepackage{multirow}
\usepackage{amsmath}
\usepackage{amsthm}
\usepackage{soul}
\usepackage{enumitem}
\usepackage{xcolor}
\usepackage{subcaption}
\usepackage{pifont}       
\usepackage{bbding}       
\usepackage{fontawesome}  
\usepackage{natbib}

\newcommand{\ours}{CausalRec}

\begin{document}

\title{CausalRec: A CausalBoost Attention Model for Sequential Recommendation}

\author{Yunbo Hou}
\authornote{Equal contribution.}
\affiliation{
	\department{School of Software and Microelectronics}
	\institution{Peking University}
	\city{Beijing}
	\country{China}
}
\email{yunboh@stu.pku.edu.cn}

\author{Tianle Yang}
\authornotemark[1]
\affiliation{
	\institution{Alibaba Group}
	\city{Beijing}
	\country{China}
}
\email{yangtianle.ytl@alibaba-inc.com}

\author{Ruijie Li}
\affiliation{
	\department{School of Software and Microelectronics}
	\institution{Peking University}
	\city{Beijing}
	\country{China}
}
\email{howtolove17@stu.pku.edu.cn}

\author{Li He}
\affiliation{
	\institution{Alibaba Group}
	\city{Beijing}
	\country{China}
}
\email{heli@taobao.com}

\author{Liang Wang}
\orcid{0009-0008-1358-9594}
\affiliation{
	\institution{Alibaba Group}
	\city{Beijing}
	\country{China}
}
\email{wangliang@taobao.com}

\author{Weiping Li}
\affiliation{
	\department{School of Software and Microelectronics}
	\institution{Peking University}
	\city{Beijing}
	\country{China}
}
\email{wpli@ss.pku.edu.cn}

\author{Bo Zheng}
\affiliation{
	\institution{Alibaba Group}
	\city{Beijing}
	\country{China}
}
\email{bozheng@alibaba-inc.com}

\author{Guojie Song}
\authornote{Corresponding author.}
\affiliation{
	\department[0]{National Key Laboratory of General Artificial Intelligence}
	\department[1]{School of Intelligence Science and Technology}
	\institution{Peking University}
	\city{Beijing}
	\country{China}
}
\email{gjsong@pku.edu.cn}

\begin{abstract}
  Recent advances in correlation–based sequential recommendation systems have demonstrated substantial success. Specifically, the attention-based model outperforms other RNN-based and Markov chains-based models by capturing both short- and long-term dependencies more effectively. However, solely focusing on item co-occurrences overlooks the underlying motivations behind user behaviors, leading to spurious correlations and potentially inaccurate recommendations. To address this limitation, we present a novel framework that integrates causal attention for sequential recommendation, CausalRec. It incorporates a causal discovery block and a CausalBooster. The causal discovery block learns the causal graph in user behavior sequences, and we provide a theory to guarantee the identifiability of the learned causal graph. The CausalBooster utilizes the discovered causal graph to refine the attention mechanism, prioritizing behaviors with causal significance. Experimental evaluations on real-world datasets indicate that CausalRec outperforms several state-of-the-art methods, with average improvements of 7.21\% in Hit Rate (HR) and 8.65\% in Normalized Discounted Cumulative Gain (NDCG). To the best of our knowledge, this is the first model to incorporate causality through the attention mechanism in sequential recommendation, demonstrating the value of causality in generating more accurate and reliable recommendations. Our code is available at \url{https://anonymous.4open.science/r/CausalRec-202B/}.
\end{abstract}



\maketitle

\input{0_intro}

\input{1_related_work}
\input{1.5_preliminary}
\input{2_method}

\input{3_experiment}

\input{GenAI}

\bibliographystyle{ACM-Reference-Format}
\balance
\bibliography{main}

\end{document}

%% file: math_commands.tex
\usepackage{amsmath,amsfonts,bm}

\def\eqref#1{equation~\ref{#1}}

\def\1{\bm{1}}

\def\vr{{\bm{r}}}

\def\vx{{\bm{x}}}

\DeclareMathAlphabet{\mathsfit}{\encodingdefault}{\sfdefault}{m}{sl}
\SetMathAlphabet{\mathsfit}{bold}{\encodingdefault}{\sfdefault}{bx}{n}

\def\gG{{\mathcal{G}}}

\def\gL{{\mathcal{L}}}

\def\gO{{\mathcal{O}}}

\def\gS{{\mathcal{S}}}

\def\gU{{\mathcal{U}}}
\def\gV{{\mathcal{V}}}

\def\sR{{\mathbb{R}}}



%% file: 0_intro.tex
\section{Introduction}
\label{sec:intro}

The goal of sequential recommendation is to uncover hidden patterns in user behavior. Current studies often incorporate correlation to identify these patterns, motivated by the natural intuition that sequential behaviors are typically related. Building on this intuition, various methods have been developed to model correlations effectively. FPMC \cite{fmpc} models user behavior as Markov chains, assuming the current behavior is influenced only by the most recent one. GRU4Rec \cite{gru4rec} employs recurrent neural networks (RNNs) to encapsulate prior actions within a hidden state. Another approach uses attention mechanisms \cite{SASRec, Bert4Rec} to differentiate the importance of historical actions.
\begin{figure}[!htbp]
	\includegraphics[width=1.0\linewidth]{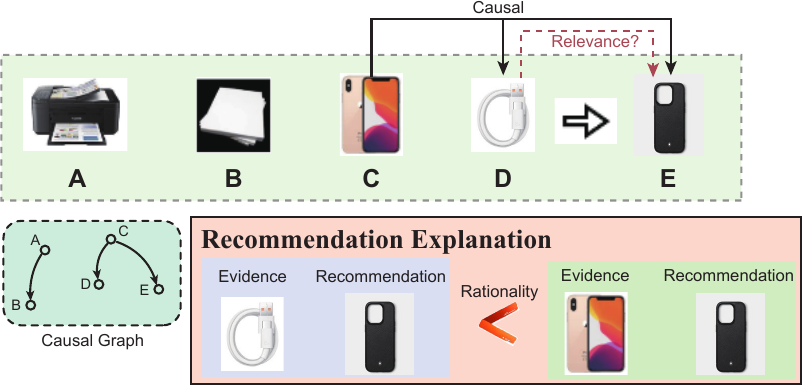}%
	\caption{Illustrating the motivation for modeling causal relationships among user behaviors. Consider the example: although the cable and the phone shell frequently co-occur in the same user behavior sequence, it does not imply that they influence each other directly. Their co-occurrence arises from their shared causal relationship with the phone.}
 \Description{}
	\label{fig:why-causal}
\end{figure}
As correlation has achieved great performance in recommendations,  these approaches typically learn \emph{spurious correlations} \cite{resort, kg}, which refers to a connection between two or more variables that appear to be causal but are not in fact. For example, in Fig. \ref{fig:why-causal}, correlation-based models might identify the correlation between the purchase of the phone shell and cable based on their co-occurrence. However, it doesn't imply that purchasing cables causes the purchase of phone shells. Instead, both purchases are driven by a shared causal factor: the purchase of phones. The example illustrates that by focusing on observed co-occurrence patterns, correlation-based models can be misled by such spurious correlations and find it hard to identify the real causal relationships, leading to wrong recommendations.

To address this limitation, we integrate causality into recommendations to capture the true causal relationships underlying user behavior.
However, incorporating it into recommendation systems poses significant challenges. First, ensuring identifiability—whether the learned causal relationship accurately reflects the user behavior pattern—is inherently difficult in sequential recommendation data. Second, the typically large number of items in recommendation systems leads to exponential growth in the search space for the learned causal relationship. Third, effectively integrating mined causal relationships into recommendation models is complex. While prior work \cite{causer} employed clustering-based methods to learn stable causal relations and filter out irrelevant items, the absence of an identifiability guarantee can risk learning incorrect user behavior patterns. The clustering-based method alleviates the huge search space issue, but it leads the a loss in user interest, and whether the reconstructed causal relation based on the cluster indeed correctly reflects the real causal relation is questionable. Additionally, the filtering methods can inadvertently discard important items, leading to a decline in recommendation performance. Our experiment in Section \ref{sec:causality-eval} has shown that our model with the filtering strategy shows a worse performance in some cases compared with SASRec, which is a correlation-based sequential recommendation model with self-attention mechanisms.

To address these challenges, we propose a CausalBoost attention model for sequential recommendation, CausalRec. We introduce a causal discovery block to infer causal relations (i.e., causal graphs) on user behavior sequences and a CausalBooster to incorporate the learned causality into the attention mechanism. We also provide theoretical guarantees for the identifiability of the learned causal graph, ensuring that it accurately captures user behavior patterns.

The causal discovery block plays a pivotal role in CausalRec. Constructing item-level causal graphs directly from real-world datasets is typically intractable because the search space grows exponentially with the number of items. To address this, we leverage the structure causal model (SCM) framework and estimate causal structures using the covariance matrix, thereby mitigating the large search space. As the attention mechanism can estimate the observed nodes of the SCM, we utilize the final layer representation of the transformer as samples to calculate covariance. By integrating layer normalization and attention mechanism in the transformer, we impose the equal variance and linear SCM assumption, ensuring our identifiability \cite{boost}.  Therefore, we can employ continuous optimization with acyclicity and sparse constraints within the transformer framework to discern causal structures from input sequences. In the CausalBooster, we address the challenge of integrating causality into recommendation models. Unlike previous methods that rely on a filtering strategy and risk discarding critical historical items, our attention-based fusion enhances the attention weight according to the learned causal graph. This approach preserves valuable user interests, even in cases where the discovered causal graph contains minor inaccuracies, ensuring robust performance. 

In summary, the main contributions of this paper are as follows:
\begin{enumerate}[nosep]
    \item We propose CausalRec, a CausalBoost attention model for sequential recommendation that integrates causality into the recommendation. To the best of our knowledge, this is the first model that incorporates causality through the attention mechanism into the sequential recommendation.
    \item We introduce a Causal Discovery Block, which successfully introduces causality into the attention mechanism with identifiability guarantees, enabling attention to focus on causally relevant signals within user interaction sequences.
    \item We present the CausalBooster, an effective procedure to utilize the causal effect matrix, integrate causality into the attention mechanism by amplifying causal effects.
    \item We conduct extensive evaluations on four datasets. \ours achieves average improvements of approximately 4.71\% in HR and 8.57\% in NDCG, validating its effectiveness in improving recommendation performance.
\end{enumerate}

%% file: 1_related_work.tex
\section{Related Work}
\label{sec:related}

This work focuses on the intersection of sequential recommendation and causal discovery. In this section, we provide an overview of recent advances in these domains and illustrate the connections between our proposed framework and existing studies.

\subsection{Sequential Recommendation}
\label{sec:related-seq-rec}

Recommendation systems have increasingly adopted models that account for the temporal and dynamic nature of user-item interactions. Traditional methods like matrix factorization \cite{mf} and static content modeling focus on aggregated user preferences, overlooking how these preferences evolve. In contrast, sequential recommendation methods take advantage of the user interaction history to predict future preferences. A common strategy is to employ encoder-like architectures that transform historical interaction sequences into latent representations, which are then used to predict the representation of the next item. 

Early sequential models, such as Markov chains \cite{fmpc}, treated user behavior as a stochastic process, relying on limited interaction windows. Extensions like Markov decision processes \cite{mdp} incorporated longer-term decision-making but struggled with data sparsity and complex transition patterns. 

Deep learning has significantly advanced this field. Recurrent neural networks (RNNs) \cite{gru4rec} captured richer temporal dependencies, while gated variants like GRUs and LSTMs improved session-based recommendations. Attention mechanisms \cite{SASRec} and transformers \cite{Bert4Rec} further enhanced performance by learning contextual embeddings without strict sequential processing, enabling them to handle complex, long-range dependencies.

Our work builds on these research directions by focusing on capturing item relationships. However, we extend them by directly discovering causal relations with identifiability guarantees, marking a significant advancement in sequential recommendation research.

\subsection{Causal Discovery}
\label{sec:related-causal-discovery}

Causal discovery seeks to uncover causal relationships among variables from observational or experimental data, forming a crucial foundation for understanding complex systems across various domains. Classical methods can be categorized into constraint-based, causal function-based, score-based, and gradient-based methods.

Constraint-based approaches rely on systematically testing for conditional independence (CI) among variables to infer causal edges under the assumption of causal sufficiency (i.e., no unmeasured confounders). A prominent example is the PC algorithm \cite{pc, PC-Stable, Copula-PC}, which starts with a fully connected graph and iteratively removes edges based on statistical CI tests until all (conditional) independencies are satisfied. While these methods are theoretically sound, they are sensitive to the quality of independence tests and struggle with high-dimensional data. 

Score-based methods cast causal discovery as an optimization problem, searching for a directed acyclic graph (DAG) that maximizes—or minimizes—a specific score function. Common scores include the Bayesian Information Criterion (BIC) \cite{bic}. While more robust to noise, these methods become computationally expensive as the number of variables increases. 

Functional causal models (FCMs) assume each variable is generated by its direct causes through a (potentially nonlinear) function combined with an independent noise term, allowing causal direction to be inferred under certain assumptions. Representative methods include LiNGAM (Linear Non-Gaussian Acyclic Model) \cite{lingam}, which exploits non-Gaussianity to discover causal ordering in linear settings, as well as additive noise models (ANM) \cite{ANM} and post-nonlinear models (PNL) \cite{PNL}, which relax linear constraints. These approaches often rely on the independence between noise and cause for identification; however, if such assumptions are violated, performance may decrease. Recent advancements in machine learning have significantly advanced causal discovery. Deep learning models, such as variational autoencoders and generative adversarial networks, have been adapted to infer causal structures in high-dimensional and non-linear datasets \cite{high-dim-causal-dis,non-linear-causal-dis}. These models exploit the representational power of neural networks to identify causal relationships that traditional methods may overlook. Differentiable causal discovery frameworks, such as NOTEARS \cite{notears,notears-mlp,notears-nonlinear, golem}, reformulate causal graph learning as a continuous optimization problem, enabling the use of gradient-based methods. These approaches have demonstrated scalability and effectiveness in handling large datasets, particularly when combined with sparsity constraints to regularize the learned causal graphs.

Our work builds on the NOTEARS framework, which formulates Bayesian structural learning as a continuous optimization problem. We incorporate this framework into recommendation systems, making an effective attempt to learn user behavior causality.

%% file: 1.5_preliminary.tex
\section{Preliminary}
\label{sec:Preliminary}
In this section, we first present the formal statement of the structural causal model and the attention mechanism, and then we establish the link between them.
\subsection{Structual Causal Models}
\label{sec:Preliminary-SCM}
A Structural Causal Model (SCM) is a framework for representing causal relationships among a set of variables. Formally, an SCM consists of:
\begin{enumerate}
\item A set of variables $x=\{x_1, x_2, \ldots, x_n\}$.
\item A directed acyclic graph (DAG) $\mathcal{G}$, where nodes represent variables and edges denote functions.
\item A set of functions $f_i$ connecting the variables, where $x_i = f_i(\text{Pa}(x_i), u_i)$, and $\text{Pa}(x_i)$ represents parents of $x_i$ in $\mathcal{G}$.
\item $u_i$: An exogenous term for each variable, assumed to be mutually independent.
\end{enumerate}
As a special case of SCMs, a linear SCM can be represented with:
\begin{equation}\label{eq-scm}
x_i =  \sum_{i \in \text{Pa}(j)} \beta_{ij} x_{j} + \lambda_{jj} u_j.
\end{equation}
The linear SCM in Equation~(\ref{eq-scm}) has the following matrix form:
\begin{equation}\label{eq-scm-m}
{X} = {B} {X} + {\Lambda} {U},
\end{equation}
where ${X}=(x_1,\dots,x_n)^\top$, ${U}=(u_1,\dots,u_n)^\top$: noise matrix, and \( {B} \in \mathbb{R}^{n \times n} \) is an edge weight matrix, an autoregression matrix, 
or a weighted adjacency matrix with each element \([{B}]_{ij} = \beta_{ij}\), in which \(\beta_{ij}\) is the linear weight of an edge from \( x_i \) to \( x_j \) and ${\Lambda}$ is a diagonal matrix with coefficient between $x_i$ and $u_i,i=1,\dots,n$.

\subsection{Attention Mechanism}
The attention mechanism, central to Transformer models \cite{transformer}, computes a weighted representation of input elements based on their relevance to the query. It linearly transformed the input ${Y}\in\sR^{n\times d}$ into three parts, $i.e.$, queries ${Q=YW_Q}\in \sR^{n\times d_k}$, keys ${K=YW_K}\in \sR^{n\times d_k}$, and values ${V=YW_V}\in \sR^{n\times d_v}$, where $n$ denotes the sequence length and $d,d_k,d_v$ represents the dimensions of inputs, queries(keys) and values. The scaled dot-product attention is applied on ${Q, K, V}$ and can be formulated as:
\begin{equation}
\text{Attention}({Q, K, V}) = \text{softmax}\left(\frac{{QK}^\top}{\sqrt{d_k}}\right){V},
\end{equation}
where the softmax function makes the rows of $\frac{{QK}^\top}{\sqrt{d_k}}$ sums to $1$. The attention mechanism enables the model to focus on the most relevant parts of the input sequence, making it highly effective for sequence modeling tasks.

\subsection{Link between Structural Causal Model and Attention Mechanism}
The previous work \cite{causal-transformer} provides a bridge between self-attention and causal discovery. To be more specific, this section describes the details of why the attention mechanism can be considered to estimate the observed nodes of a SCM. 
{\color{red} }

Initially, we explain that the covariance over the outputs of the attention has a similar formulation to the covariance over observed nodes in an SCM. For a linear SCM, ${X}={BX}+ {\Lambda U}$, then 
\begin{equation}\label{eq1}
    {X = (I-B)}^{-1}{\Lambda U}
\end{equation}
where equation~(\ref{eq1}) represents a system with the outputs $X$, inputs $U$ and weights ${(I-B)}^{-1}$. The covariance matrix of the output is
\begin{equation}\label{eq2}
\begin{aligned}
    Cov({X})&=\mathbb{E}[({X}-\mu_{X})({X}-\mu_{X})^\top]\\
    &={({I}-{B})}^{-1}{\Lambda} Cov({U}) ({({I}-{B})}^{-1}{\Lambda})^\top,
\end{aligned}
\end{equation}
where $\mu_{X}={({I}-{B})}^{-1}{\Lambda}\mu_{U}$, and $Cov(\cdot)$ means the covariance matrix.

On the other hand, an attention layer estimates the attention matrix ${A}$ and a values matrix ${V}$ from embeddings ${Y}$. The output embeddings are ${Z}={AV}$
If we view ${V}$ as a random variable with mean ${\mu_V}$ and covariance $Cov({V})$, then the output ${Z}$ also has a covariance matrix given by
\begin{equation}\label{eq3}
\begin{aligned}
    Cov({Z})&={\mathbb{E}[({Z}-\mu_{Z})({Z}-\mu_{Z})^\top]}={A}Cov({V}){A}^\top,
\end{aligned}
\end{equation}
where ${\mu_Z}={A\mu_V}$. Comparing this to $Cov({X})$ in equation~\eqref{eq2} shows a striking similarity: in each case, the output covariance is formed by applying a linear transformation (either ${({I}-{B})}^{-1}{\Lambda}$ in the SCM or ${A}$ in the attention mechanism) to the input covariance ($Cov({U})$ or $Cov({V})$, then multiplying by its transpose.

Hence, the learned weight ${A}$ in the attention mechanism plays the same linear mapping role as ${({I}-{B})}^{-1}{\Lambda}$ in SCM. In the SCM, the term ${\Lambda}$ captures how exogenous inputs ${U}$ influence the system, and ${({I}-{B})}^{-1}{\Lambda}$ encodes the dependence among observed variables. Analogously, in the attention mechanism, ${V}$ represents the content being passed among tokens (similar to “inputs”), while ${A}$ acts as a learned “adjacency matrix” that decides how different tokens interact with each other. Consequently, the attention mechanism can be viewed as estimating, in a data-driven way, relationships among observed nodes—mirroring how ${({I}-{B})}^{-1}{\Lambda}$ describes the relationships in a linear SCM.

%% file: 2_method.tex
\section{Method}
\label{sec:method}

The overview of CausalRec is depicted in Fig. \ref{fig:causalrec}. The proposed model integrates a Causal Discovery Block to identify causal relations. To incorporate causality into the recommendation and further enhance sequential recommendation performance, we propose a CausalBooster that incorporates the discovered causal relations into the attention mechanism. Consistent with previous sequential recommendation methods, we include an Embedding Layer and a Prediction Layer to learn item embeddings and generate logits for the output (see details in Appendix \ref{sec:app-emb-pred}). In summary, CausalRec comprises four main modules: the Embedding Layer, Causal Discovery Block, CausalBooster, and Prediction Layer.

\begin{figure}[!htbp]
	\includegraphics{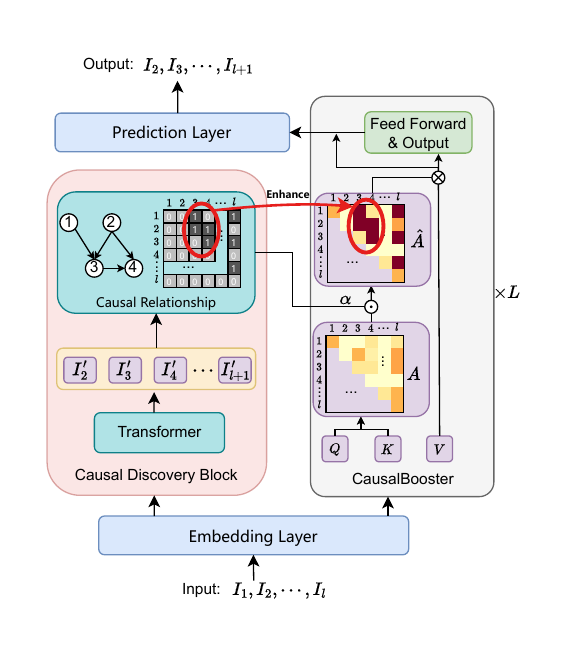}
	\caption{Architecture of our proposed CausalRec. We propose a Causal Discovery Block to learn user behavior causality and incorporate it into our model with CausalBooster.}
 \Description{}
	\label{fig:causalrec}
\end{figure}

\subsection{Causal Discovery Block}
\label{sec:causal-discover}

Identifying causal relations in recommendations faces identifiability concerns, which denote that multiple SCMs can produce the same observed data distribution, making it impossible to single out one true causal structure. Previous work, such as the causer \cite{causer}, has attempted to solve the issues by clustering methods, but failed. To overcome these limitations, we propose an Item-level Causal Discovery Block and provide the identifiability proof.

\subsubsection{Item-level Causal Discovery Block}

Different from the causer, we directly mine the item's causal relationship. We leverage the SCM and estimate causal structures with a covariance matrix:


\begin{equation}\label{eq:cov}
    Cov(i,j) = \frac{1}{N}\sum_{k=1}^Nx_{ki}x^\top_{kj}\in\mathbb{R}^{n\times n},
\end{equation}
where $x_{i} \in \sR^n$ is the vector of items in a user’s sequence to avoid confusion with the identity matrix I symbol, $n$ is the input sequence length, and $N$ denotes the batch size. To address the identifiability issue, we use the final layer representation of the transformer as samples to calculate covariance. As illustrated in Section \ref{sec:Preliminary}, the attention mechanism can estimate the linear SCM observed nodes, and the layer normalization mechanism naturally imposes an equal noise variance assumption across different dimensions. Since $X$ follows linear SCM equations $X=BX+\Lambda U$, properties of the exogenous variables $U$—namely independence and finite variance—enable us to express $Cov(X)$ in terms of $B$ and $\Lambda$. Therefore, $B$ can be estimated using $Cov(X)$. Additionally, the incorporation of the transformer introduces the linear SCM and equal noise variance assumptions. These additional assumptions ensure that, with acyclicity and sparse constraints on $Cov(X)$ in continuous optimization, the DAG within the data generation process is theoretically identifiable \cite{boost}. This is due to the attention mechanism and the property of Layer Normalization, thereby resolving the identifiability issue in prior methods. 

\subsubsection{Identifiability Analysis}

In this section, we provide the causal identifiability (DAG can be uniquely determined) of the Causal Discovery Block. As \cite{causal-transformer} provides a bridge between self-attention and causal discovery, and \cite{transformer} demonstrates that the transformer can learn causal structure with gradient descent, here we show the identifiability of the outputs from the attention mechanism with the unique property of Layernorm.
To the beginning, we provide the causal identifiability of linear SEM as Lemma~\ref {lm1}: if the exogenous variables have equal variances, then the directed acyclic graph of this linear SEM can be uniquely identifiable. Here we denote $\mathcal{G}(B)$ as directed graph with adjacent weight matrix $(\beta_{j,k})B\in\mathbb{R}^{n\times n}$, vertex set $V=\{1,\dots,n\}$ and edge $E(B)$ be the support of $B$ ($E(B)=\{(k,j):\beta_{j,k}\neq 0\}$). Let $X\sim(B_X,\sigma^2_X)$ denote the exogenous variables of linear SEM that have equal variance $\sigma^2_X$.

\begin{lemma}[Identifiability Conditions for Linear SEMs \citep{Chen_2019,JMLR-Park}]\label{lm1}
Let \( P(X), P(Y) \) be generated from a linear SEM (\ref{eq-scm}) with DAG \( \mathcal{G}(B_X), \mathcal{G}(B_Y) \) and true ordering \( \pi_X, \pi_Y \), $X\sim(B_X,\sigma^2_X)$, $Y\sim(B_Y,\sigma^2_Y)$ with both \( G_X, G_Y \) directed and acyclic. If $\text{var}(X)=\text{var}(Y)$, then $\mathcal{G}(B_X)=\mathcal{G}(B_Y)$, $B_X=B_Y$ and $\sigma^2_X=\sigma^2_Y$.
Then, DAG \( G \) is uniquely \textit{identifiable} if exogenous variables have equal variances.
\end{lemma}

A critical insight is that the self-attention mechanism, particularly in our context, can be formulated as a linear SEM, thus satisfying the preconditions of Lemma~\ref{lm1}. Specifically, the Layer Normalization step ensures that the exogenous variables (derived from input embeddings) have identical variance. This allows us to formally state the identifiability of the learned causal structure.

\begin{proposition}{(Causal Identifiability of Self-Attention with LayerNorm)} Let $P(Z)$ be generated from $Z=BZ+DV$ with pre-Layer normalization of diagonal matrix $D=\text{diag}(\frac{1}{\|V_1\|_2},\dots,\frac{1}{\|V_n\|_2})$ to make exogenous variables $\tilde{V}=DV$ with equal variances, $\mathcal{G}(Z)$ can be uniquely identified.
\end{proposition}

\begin{proof}
The output $Z$ of the self-attention layer is a linear transformation of its input values $V$. For each output $Z_i$, we have:
$Z_i = \sum_{j=1}^{n} \text{softmax}\left(\frac{Q_i K_j^T}{\sqrt{d_k}}\right) V_j = \sum_{j=1}^{n} A_{ij} V_j$,
where $A$ is the attention matrix. This confirms that the relationships among the components of $Z$ are linear.

We can re-express this relationship with a linear SEM. For simplicity and without loss of generality, let's consider the case where the embedding dimension of $V$ is 1. The structural model for the item representations $Z$ can be written as $Z = BZ + \tilde{V}$, where the input value matrix $V$ serves as the basis for the exogenous variables. The pre-Layer Normalization, represented by the diagonal matrix $D$, transforms $V$ into $\tilde{V} = DV$. This normalization ensures that each component $\tilde{V}_i$ has unit variance, making them independent and identically distributed (i.i.d.) noise terms. By assuming that the causal graph $\mathcal{G}(B)$ is a DAG, the matrix $(I-B)$ is invertible. This allows for a unique solution for $Z$:
$Z = (I-B)^{-1}\tilde{V}$
This equation describes a linear SEM where the output $Z$ is generated from exogenous variables $\tilde{V}$ that have equal variances. According to Lemma~\ref{lm1}, the underlying DAG $\mathcal{G}(B)$ is therefore uniquely identifiable from the covariance matrix of $Z$, which is given by $Cov(Z) = (I-B)^{-1} Cov(\tilde{V}) ((I-B)^{-1})^\top$. This completes the proof.
\end{proof}
In summary, the combination of self-attention (with causal masking ensuring the DAG property) and Layer Normalization satisfies the conditions for causal identifiability in linear SEMs. The autoregressive structure of the attention matrix $A$ determines the causal ordering, and LayerNorm ensures the homogeneity of exogenous noise variances, thereby guaranteeing that the causal graph over the output representations $Z$ is identifiable. Then, once we find a DAG during the training procedure, we say that the causal relationship matrix $R$ in section~\ref{sec:causal-layer} is uniquely identifiable.

To find a DAG, we introduce a general method to add acyclicity and sparse constraints in our optimization as NOTEARS~\cite{notears-nonlinear}. One can also apply other acyclicity constraints which has less time complexity. The loss consists of two components:(1) the acyclicity constraint, which is to ensure that the learned causal graph is a directed acyclic graph. (2) Sparse Regularization, which is to encourage sparsity in causal relationships. The acyclicity constraint comes from the taylor expansion $trace(e^A)$$=trace(I)$$+trace(A)$$+trace(A^2)$$+\cdots$, where $A=W\odot W,W\in \sR^{n\times n}$. Suppose $W$ represents the adjacency matrix corresponding to graph $\gG$ with $n$ nodes and $(A^k)_{ij}$ denotes the number of k-step paths from node $i$ to node $j$. Then $n=trace(I)=trace(A)+trace(A^2)+\cdots$ represents that there is no path from node $i$ to node $i$, which indicates the graph $\gG$ is a DAG. Therefore, the acyclicity constraint can be written as below:
\begin{equation}
    \text{trace}(e^{W \odot W})=n,
\end{equation}
where $W$ denotes the covariance matrix $Cov(X)$, $\odot$ represents the Hadamard product, and $n$ is the input sequence length. The idea of sparse regularization comes from the fact that in the user behavior sequence, the causal item corresponding to the target item should be sparse. Therefore, we incorporate the $L_1$ matrix norm on the covariance matrix as our sparse regularization to encourage sparse causal relations among items.



\subsection{CausalBooster}
\label{sec:causal-layer}

The CausalBooster is designed to integrate causality into the sequential recommendation. It achieves this by stacking multiple CausalBoost Attention (CBA) Layers, each consisting of a CBA Layer and a feed-forward network. This structure allows the model to enhance relevant behaviors without discarding important information, addressing key challenges in incorporating causal relations.

\subsubsection{CausalBoost Attention Layer} 


Previous approaches \cite{causer} would filter out items with low causal relevance based on a learned causal matrix. However, that risks discarding user interest in sequences.

Instead, we introduce a multiplicative enhancement of attention. The original attention is calculated by $A = \text{softmax}(\frac{QK^\top}{\sqrt{d}})$ and $\mathrm{Z}=\mathrm{AV}$. In Section \ref{sec:causal-discover}, the causal relationships are derived from $B$ where $X=BX+\Lambda U$ denotes the linear SCM among items (also corresponding to $V$). Therefore, we form an adjusted attention matrix to apply the causal relation matrix on the unnormalized attention weight. This ensures no critical behavior is outright discarded and achieves numerical stability for smooth training:
\begin{equation}
    \Tilde{A}^l = A^l \odot (\mathbf{1}_n\mathbf{1}_n^\top + \alpha R  ),
\end{equation}
where $\odot$ denotes the Hadamard product or element-wise product, $A^l \in \sR^{n_{max} \times n_{max} }$ is the attention matrix of the $l_{\text{th}}$ CBA Layer, $\mathbf{1}_n:=\{1,\dots,1\}^\top\in\mathbb{R}^n$ means the all-ones vector and $\alpha$ is a scalar hyperparameter controlling how strongly causal relationships influence the attention. $R \in \sR^{n\times n}$ denotes the learned causal relationship matrix (using the method in Section \ref{sec:causal-discover}) where $1$ denotes the items have a causal relation and $0$ denotes not. We then apply a standard attention softmax and the prefix mask: 
\begin{equation}
    Z^l = \text{softmax}(\mathcal{M} + \Tilde{A}^l)V^l,
\end{equation}
where $V^l = X^l W^l_V$ represents the values in the attention mechanism, with $X^l$ as the input to the $l_\mathrm{th}$
$$
\mathcal{M}(x,y)=\begin{cases}
	0&if\;\;\;\;x\le y,		\\
	-\infty &otherwise.		\\
\end{cases} 
$$
\subsubsection{Feed-Forward Network and Output Layer} 
To enrich the capabilities of the model's representation, a two-layer point-wise feed-forward network is incorporated after each CBA Layer. This network introduces nonlinearity and facilitates interaction across latent dimensions. It is formulated as:
\begin{equation}
    \hat{X}^l = (\text{ReLU}(\Tilde{X}^lW^l_1+b^l_1))W^l_2+b^l_2,
\end{equation}
where $W^l_1,W^l_2 \in \sR^{D \times D}$ are the learnable parameter matrix and $b^l_1,b^l_2 \in \sR^{D}$ is the learnable parameter vector. Following the transformer architecture \cite{transformer}, we apply residual connections, layer normalization, and dropout layers to alleviate overfitting. The output of the layer is defined as:
\begin{equation}
    X^{l+1}=\text{LayerNorm}(X^l + \text{Dropout},(\hat{X}^l)),
\end{equation}
where the Dropout and Layer Normalization are defined as:
\begin{equation}
\begin{aligned}
    \text{Dropout}(\vx) &= \vr \odot \vx, \\
    \text{LayerNorm}(\vx)&=\theta_1 \odot \frac{\vx-\mu}{\sqrt{\sigma^2+\epsilon}} + \theta_2, 
\end{aligned}
\end{equation}
where $\mu$ and $\sigma$ are the mean and variance of $\vx$, $\theta_1$ and $\theta_2$ are learned scaling factors and bias terms. $\vr$ is the random vector and $\vr_i \sim \text{Bernoulli}(p)$ with probability parameter $p$.



\subsection{Embedding Layer and Prediction Layer}
\label{sec:app-emb-pred}

\subsubsection{Embedding Layer}


Following previous practices, we first truncate the given user behavior interaction sequence $o={I_1, I_2,\cdots, I_{n}}$ by remove the early item $I_i, i > n_{max}$ and pad empty items for a short sequence $o_j,n < N_{max}$ to obtain fixed sequence set $\gO=\{u_k, I_1^k, I_2^k, \cdots, I_{N_{max}}^k\}_{k=1}^n$, where $n_{max}$ denotes the maximum sequence length. We use an item embedding matrix $M \in \sR^{|\gV|\times D}$, where $D$ denotes the hidden representation dimension, to define the embedding of the sequence $E^k=M_{o_k}$. To make our sequence more sensitive to the position of the sequence, we define the positional embedding matrix $P\in \sR^{N_{max} \times D}$ and add it to the sequence embedding. Our embedding layer is written as below:
\begin{equation}
    E^{k} = \text{Dropout}(M_{o_k}+P)\ .
\end{equation}

\subsubsection{Prediction Layer}

The Prediction Layer serves as the final component of the CausalRec. Suppose the user’s preference score for an item $i$ is calculated based on the interaction sequence processed by the CausalBooster. The calculation is performed using the dot product between the embedding of item $i$ and the output of the Encoder, which quantifies the similarity between the user’s preference and the item’s representation. It can be defined as below:
\begin{equation}
    \hat{y}_i=e^T_iI^L_{n_k},
\end{equation}
where $\hat{y}$ represents the output logits, $e_v$ denotes the representation of item $i$ which comes from the item embedding matrix $M$ and $I^L_{n_k}$ is the output representation of $L_{\rm{th}}$ CausalBooster layer.

\subsection{Training Loss}

Given the training, user interaction set $\gO_{\text{train}} = \{u_k, I_{1}^{k}, I_2^k, \cdots$, $ I_{n_k}^k\}_{k=1}^{n}$, the primary goal is to mine causal relations to enhance recommendation performance. To achieve this, the loss function is composed of two components: (1) Recommendation Loss ($\gL_{\text{rec}}$): Captures the model's ability to predict the next item in a user sequence. (2) DAG Constraint ($\gL_{DAG}$): Maintaining sparsity and acyclicity constraints in the learned causal graph. The final loss function is as below: 
\begin{equation}
    \gL = \gL_{rec} + \lambda * \gL_{L_1} + \gL_{DAG},
\end{equation}
where $\lambda$ is the penalty coefficient for the $L_1$ sparsity term. This unified loss ensures that the model learns accurate recommendations while simultaneously identifying reliable causal relationships.

\subsubsection{Recommendation Component}

For the recommendation component, we treat the sequential recommendation as a next-item prediction problem and use a cross-entropy loss function, which is widely adopted in related tasks \cite{ce1,ce2,ce3}. It is defined as:
\begin{equation}
    \begin{aligned}
\gL_{rec} = -\sum_{o_k \in \gO_{train}}\sum_{i}^{n_k}\sum_{s \in \gS}[ y_{i,s} log (\sigma(\hat{y}_{i,s})) + \\
(1-y_{i,s})log(1-\sigma(\hat{y}_{i,s})) ],
\end{aligned}
\end{equation}
where $\sigma(\cdot)$ denotes the sigmoid function, $\gS$ is the item set, and $y_{i,s}$ means whether item $s$ is the next item in the user sequence $u_k$.
\subsubsection{DAG Component}
As described in Section \ref{sec:causal-discover}, the DAG component consists of two parts: an acyclicity constraint and an $L_1$ sparse penalty. For the user $u_k$ with its learned causal $W_k$, incorporating the L1 penalty $\gL_{L1}=\sum_{o_k \in \gO_{train}} ||W_k||_1$ is simple. However, integrating the acyclicity constraint is the opposite, as the acyclicity constraint $\text{trace}(e^{W \odot W})=n$ is nonconvex. Hence, followed by the NOTEARS, we transform it into an unconstrained subproblem:
\begin{equation}
    \gL_{DAG} =  \sum_{o_k \in \gO_{train}} \frac{\rho}{2}|h(W_k)|^2 + \beta |h(W_k)|,
\end{equation}
where $h(W_k)=\text{trace}(e^{W_k \odot W_k})-n$, and $W_k$ is the learned causal graph for user $u_k$. The $\beta$ follows the rule $\beta \leftarrow \beta + \rho \kappa$ to update after each epoch, where $\kappa=mean_{o_k \in \gO_{train}} h(W_k)$. Following NOTEARS \cite{notears}, we set $\rho$ update rule $\rho \leftarrow \rho * \gamma_1$ if $\kappa \geq \gamma_2 \kappa^-$ after each epoch, where $\kappa^-$ denotes the $\kappa$ in the last epoch and initially set to 0.

\subsection{Complexity Analysis}
\subsubsection{Model Complexity}
The CausalBoost model is based on the Transformer architecture. Its computational complexity is 
$O(bl^2d+bd)$, where the first term corresponds to the attention mechanism’s complexity and the second term corresponds to the feed-forward network (FFN) complexity. The variables $b$, $l$, and $d$ represent the batch size, sequence length, and hidden dimension, respectively.
\subsubsection{Loss Complexity}
The primary contributor to the loss computation complexity is the acyclicity constraint, which relies on the matrix exponential. Prior methods such as NOTEARS\cite{notears}, NOFEARS\cite{nofears}, and NOBEARS\cite{nobears} (see Table \ref{tab:dag-complexity}) introduce optimizations to reduce computation. Our approach follows NOTEARS to compute the constraint, which has the worst complexity of $O(l^3)$. This can achieve the theoretically optimal complexity with our covariance matrix, Eq.(\ref{eq:cov}). As $l$ is set to $200$ and $b$ is set to $256$, the constraint complexity is similar to the model complexity, and compared with SASRec, our runtime increases by only $0.35$s per epoch on average. It should be noted that our approach can be further optimized by blocking the matrix with a permutation matrix and computing constraints on each block matrix, whose complexity is $O(l^2)$ or $O(\sum_i^m r_i^3)=O(r^3)$. $r_i$ is the rank of the block matrix and $\sum_i^m r_i=l,\ r_i < l$ with the biggest rank $r=\mathrm{max} \{r_i\}_{i=1,\dots,m}$, and the number of blocks $m$. The time complexity comparison among our approach, NOTEARS, and its variants is shown in the Table~\ref{tab:dag-complexity}. 

\begin{table}[htbp]
\caption{Time complexity of different acyclicity constraints (worst / average / best).}
\label{tab:dag-complexity}
\centering
\resizebox{0.95\linewidth}{!}{
\begin{tabular}{l|c|c|c|c}
\hline
Case & NOTEARS & NOBEARS & NOFEARS & Ours (covariance matrix) \\
\hline
Worst   & $O(\ell^{3})$ & $O(\ell^{3})$ & $O(\ell^{3})$ & $O(\ell^{3})$ \\
Average & $O(\ell^{3})$ & $O(\ell^{2})$ & $O(\ell^{3})$ & $O(\ell^{3})$ \\
Best    & $O(\ell^{3})$ & $O(\ell^{2})$ & $O(\ell^{3})$ & $O\!\left(r^{3}\right)$ \\
\hline
\end{tabular}}
\end{table}
\subsubsection{Overall Complexity}
As above, the total best complexity mainly comes from the attention mechanism: $O(bl^2d+bd+r^3)=O(bl^2d)$.

%% file: 3_experiment.tex
\section{Experiments}
\label{sec:exp}

\begin{table*}[htbp]
\caption{{Overall comparison between baselines and our models. The best performance is highlighted in bold, and the suboptimal results are shown with a dashed line below. All the numbers are percentage values with "\%" omitted.}}
\label{tab:main-result}
\centering
\resizebox{0.9\linewidth}{!}{
\begin{tabular}{l|cc|cc|cc|cc}
\hline
Datasets & \multicolumn{2}{c|}{Movielen-1m} & \multicolumn{2}{c|}{Foursquare} & \multicolumn{2}{c|}{LastFM}     & \multicolumn{2}{c}{KGRec-music} \\ \hline
Metric@10 & NDCG$\uparrow$  & {HR$\uparrow$}    & NDCG  & {HR}    & NDCG  & {HR}    & NDCG           & HR             \\ \hline
BPR         & 7.33  & {16.34} & 14.20 & {29.09} & 11.22 & {17.00} & 16.54          & 33.03          \\ \hline
GRU4Rec     & 6.45  & {13.54} & 22.87 & {38.69} & 4.80  & {10.46} & 14.61          & 22.43          \\ \hline
STAMP       & 21.46  & {{42.16}}  & 17.41  & {35.18}  & 3.50  & {7.06}  & 38.27           & 73.68           \\ \hline 
Causer      & 6.51  & {14.07} & 4.37  & {10.16} & 5.04  & {10.78} & 12.53          & 22.95          \\ \hline
VTRNN       & 19.88 & {44.40}  & 21.09 & {30.93} & 4.74  & {9.82}  & 26.97          & 47.64          \\ \hline
SASRec      & 59.18  & \underline{82.25} & 20.75  & {39.12} & \underline{13.57}  & {18.95} & \underline{75.37}          & \underline{93.21}          \\ \hline

BSARec      & \underline{60.75} & {78.61} & \underline{24.45} & \underline{39.24} & 13.46 & \underline{19.03} & 74.55          & 93.13          \\ \hline
Ours &
  \textbf{72.40\textsubscript{{\scriptsize\textbf{$\uparrow$ 11.65}}}} &
  {\textbf{88.94}\textsubscript{{\scriptsize\textbf{$\uparrow$ 6.69}}}} &
  \textbf{39.94\textsubscript{{\scriptsize\textbf{$\uparrow$ 15.49}}}} &
  {\textbf{53.89}\textsubscript{{\scriptsize\textbf{$\uparrow$ 14.65}}}} &
  \textbf{16.36\textsubscript{{\scriptsize\textbf{$\uparrow$ 2.79}}}} &
  {\textbf{24.00}\textsubscript{{\scriptsize\textbf{$\uparrow$ 4.97}}}} &
  \textbf{79.70\textsubscript{{\scriptsize\textbf{$\uparrow$ 4.33}}}} &
  \textbf{95.77\textsubscript{{\scriptsize\textbf{$\uparrow$ 2.56}}}} \\ \hline
\end{tabular}
}
\end{table*}

In this section, we present comprehensive experiments to evaluate the effectiveness of \ours{} in sequential recommendation tasks. We first outline the experimental setup and describe the baseline models in Section \ref{sec:exp-detail}. Next, we compare \ours{} with existing baselines to assess its performance. Then, to validate the role of causality in \ours{}, we conduct both quantitative and qualitative experiments, providing further insights into the advantages of incorporating causality into sequential recommendations.

\subsection{Experiment details}
\label{sec:exp-detail}
This subsection outlines the datasets, baselines, and implementation details used in our experiments. These components provide a comprehensive foundation for evaluating the performance of \ours{} in sequential recommendation tasks.

\subsubsection{Dataset}
\label{sec:dataset}
We conduct our sequential recommendation experiments on the following real-world datasets: Movielens-1M\footnote{https://grouplens.org/datasets/movielens/} is a popular movie recommendation dataset collected from GroupLens Research, which contains user ratings on movies. LastFM\footnote{http://millionsongdataset.com/lastfm/} is a music recommendation dataset that contains user interaction with music, such as artist listening records. Foursquare\footnote{https://www.kaggle.com/datasets/chetanism/foursquare-nyc-and-tokyo-checkin-dataset} is a location-based recommendation dataset including user check-ins of restaurants in Tokyo for about 10 months. KGRec-music\footnote{https://www.upf.edu/web/mtg/kgrec} is a music recommendation dataset collected from songfacts.com and last.fm websites. Table \ref{tab:dataset-info} summarizes the statistical information of the above datasets.

\begin{table}[htbp]
\caption{ Statistics of the datasets, where the "SeqLen" denotes the average sequence length of each user.}
\label{tab:dataset-info}
\centering
\resizebox{1.0\linewidth}{!}{
\begin{tabular}{l|c|c|c|c|c}
\hline
Dataset     & \#User & \#Item & \#Interaction & \#Sparsity & \#SeqLen \\ \hline
Movielen-1m & 6040   & 3952   & 1000209       & 95.81\%  & 163.5      \\ \hline
Foursquare  & 1083   & 38333  & 91024         & 99.78\%  & 82.05      \\ \hline
LastFM      & 1892   & 17632  & 92834         & 99.72\%  & 47.08      \\ \hline
KGRec-music & 5199   & 8640   & 751531        & 98.32\%  & 142.55     \\ \hline
\end{tabular}
}
\end{table}

\subsubsection{Baselines}
\label{sec:baselines}
To evaluate the effectiveness of our model, we compare it against well-known baselines:
\begin{itemize}
    \item BPR: BPR \cite{fmpc} is a well-known recommendation model for capturing user implicit feedback. It is combined with matrix factorization to model user-item preferences.
    \item GRU4Rec: GRU4Rec \cite{gru4rec} is a sequential recommendation model based on gated recurrent units. It leverages sequential user interactions to predict the next item.
    \item STAMP: STAMP \cite{stamp} is a sequential recommendation model that emphasizes short-term user preferences while integrating long-term memory. It uses attention mechanisms to capture recent interactions.
    \item Causer: Causer \cite{causer} is a sequential recommendation model incorporating the learned causality on user behavior sequences in the clustering level. Identifying the clustering causal relationships enhances the understanding of user intent and improves recommendation accuracy.
    \item VTRNN: VTRNN \cite{vtrnn} is a sequential prediction model that combines visual and textual features using a recurrent neural network to capture multimodal contextual information.
    \item SASRec: SASRec \cite{SASRec} is a recommendation model based on self-attention mechanisms. It captures long-term dependencies in user behavior by modeling user interactions.
    \item BSARec: BSARec \cite{bsarec} is a sequential recommendation model that introduces an attentive inductive bias to enhance predictions beyond traditional self-attention mechanisms.
\end{itemize}

\begin{figure*}
    \centering
    \includegraphics{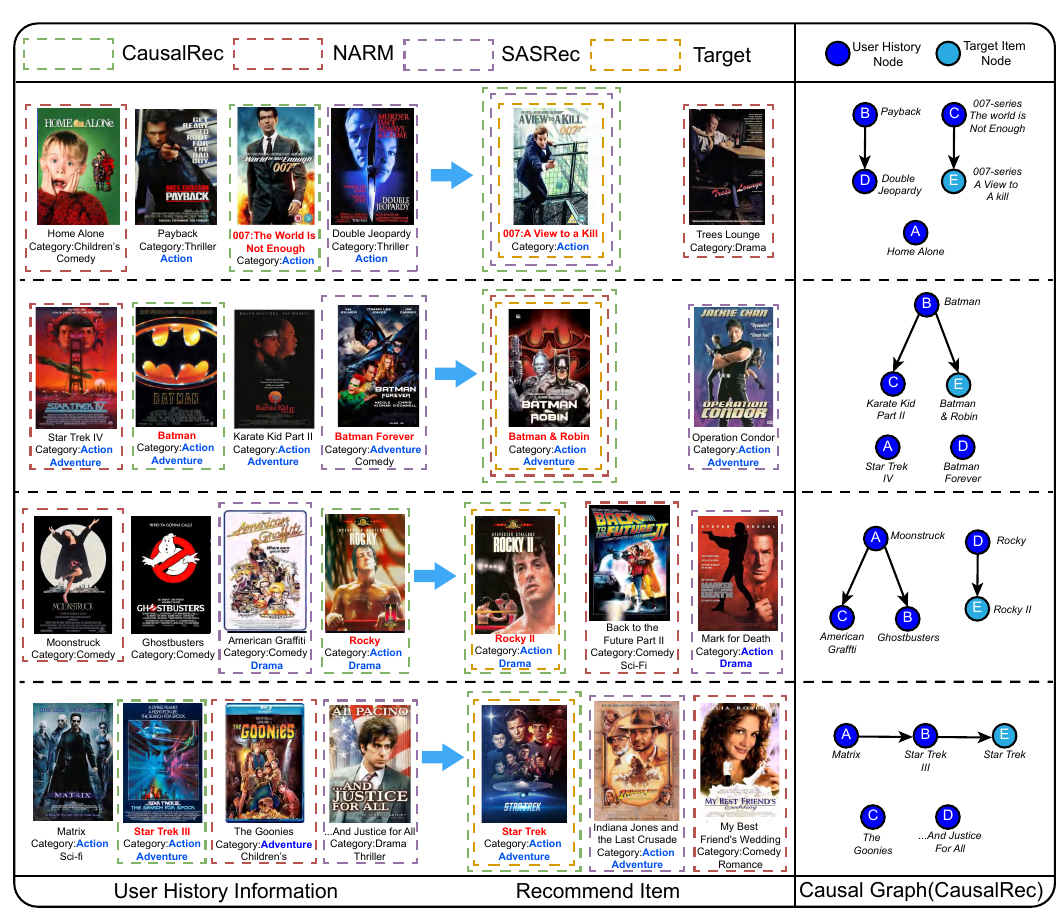}
    \caption{Visualization of user interaction histories (left), recommended items (center), and CausalRec’s discovered causal graph (right). The film title in red on the left belongs to the same series as the target item in the middle (also shown in red), which is the user’s actual choice that the model seeks to recommend. Categories shown in blue indicate the target item’s category. Different colored boxes denote explanations from CausalRec (green), NARM (red), and SASRec (purple), while the orange box highlights the ground‐truth (target) item.}
    \label{fig:enter-label}
    \Description{}
\end{figure*}

\subsubsection{Implementation details}
\label{rec:implem-details}
In our experiments, we first organize the interactions of each user according to the timestamp. Then, following the common practice, we use the last and the second last interactions of each user behavior sequence as validation sets and testing sets, while others are used as training sets. For our architecture, we use two CBA layers for the CausalRec. The optimizer is the Adam optimizer, the learning rate is $0.001$, and the batch size is $256$. The dropout rate for all datasets is set to $0.2$, and the maximum sequence length is set to $200$ for all datasets. The hyperparameters, $\alpha,\lambda$ are selected based on grid search, and the ranges are $\{10^{-8}, 10^{-7}, 10^{-6}, \cdots, 10^6, 10^7, 10^8\}$. Others in baseline models are set to the default value in the original paper.

For the evaluation, we take the widely used metrics, including HR and NDCG, for model evaluation. Specifically, suppose $A_u$ and $B_u$ are the set of items recommended to user $u$ and the ones interacted with by them in the testing set. $Z$ is the number of recommended items. $R(i)$ is the relevance score, where $R(i)=1$ denotes the $i$th item belongs to $B_u$, otherwise $R(i)=0$. Then the formulas for computing HR and NDCG are:
\begin{equation}
    \text{HR@Z}=\frac{1}{|\gU|} \sum_{u\in\gU} \frac{\sum_{i=1}^Z R(i)}{|B_u|}
\end{equation}
\begin{equation}
    \mathrm{DCG_u@Z}=\sum_{i=1}^Z\frac{R(i)}{\log_2(i+1)}
\end{equation}
\begin{equation}
    \text{NDCG@Z}=\frac{1}{|\gU|} \sum_{u\in\gU}\frac{\mathrm{DCG_u@Z}}{\max_{u\in\gU}(\mathrm{DCG_u@Z})}.
\end{equation}
In our experiment, we set $Z=10$. To avoid heavy computation on all user-item pairs, we follow the strategy in \cite{SASRec}. For each user, we sample $100$ negative items and rank them with the ground truth.

\input{tab_abla}

\subsection{Recommendation Performance}
\label{sec:rec-performance}

Table \ref{tab:main-result} presents the results across four datasets. Our model achieves superior performance on all datasets, demonstrating average improvements of {$8.56\%$} in NDCG@10 and ${7.21\%}$ in HR@10 over the best baseline. Notably, the Foursquare dataset shows ${15.49\%}$ and ${14.65\%}$ improvements in NDCG and HR. These results validate the effectiveness of incorporating item-level causality into sequential recommendations, enhancing target item prediction by focusing on item causal relations. While Causer also introduces causality, its suboptimal performance stems from: (1) less effective RNN-based architectures for long sequences compared to our Transformer-based model; (2) learning pseudo-causal relationships at the cluster level from side information, which poorly represents sparse and complex real item causal relationships; and (3) a lack of identifiability in its learned causal relationships, leading to history information loss when filtering items.

\subsection{Causality Evaluation}
\label{sec:causality-eval}

To evaluate the effect of causality and correlation in improving recommendation system performance, we conduct the ablation study below. Table~\ref{tab:abla} shows the results, where CausalRec(w/o Attention) denotes we remove the attention mechanism, relying solely on the causal relation matrix, CausalRec(w/o Causality) represents we remove the causal relation matrix, making it equivalent to SASRec, CausalRec(w/o sparse) denotes removing the $L_1$ constraint in the loss function, and CausalRec(w/filter) in Table~\ref{tab:abla-filter} denotes we use a filtering strategy instead of an enhancing strategy. The filtering strategy of CausalRec(w/filter) can be written as follows:
\begin{equation}
    \Tilde{{A}}^l = \text{softmax}({A}^l + \mathcal{M}_{R})
\end{equation}
where all symbols are defined in Section \ref{sec:causal-layer}. The $\mathcal{M}_{R}$ is defined as:
\begin{align}
\mathcal{M}_{R}(x,y)=\begin{cases}
	-\infty&if\;\;\;\; {R}_{x,y}\le threhold		\\
	0 &otherwise		\\
\end{cases} 
\end{align}
where $threhold$ is set to $0.9$.

From the results, CausalRec(w/o Attention) underperforms CausalRec(w/o Causality) on KGRec-music, though causality still shows an advantage over correlation on the other three datasets. CausalRec(w/o Sparse) shows significant performance drops on Movielens-1M and Foursquare, emphasizing the necessity of sparse constraints. CausalRec(w/filter) improves performance on three datasets but significantly declines on KGRec-music (even below SASRec), indicating that improper fusion can degrade performance.
\begin{table}[!t]
\caption{Comparison of filter strategy and Causalbooster in terms of NDCG@10 and HR@10. All the numbers are percentage values with "\%" omitted.}
\label{tab:abla-filter}
\centering
\begin{tabular}{cc|c|c}
\hline
\multicolumn{2}{c|}{Model}                                &       &       \\ \cline{1-2}
\multicolumn{1}{c|}{Dataset} & {Metric@10} & \multirow{-2}{*}{CausalRec(w/filter)} & \multirow{-2}{*}{CausalRec} \\ \hline
\multicolumn{1}{c|}{}                              & NDCG & 68.05\textsubscript{\textcolor{purple}{\scriptsize\textbf{$\downarrow$ 4.35}}} & \textbf{72.40} \\ \cline{2-4} 
\multicolumn{1}{c|}{\multirow{-2}{*}{Movielen-1m}} & HR   & 85.66\textsubscript{\textcolor{purple}{\scriptsize\textbf{$\downarrow$ 3.28}}} & \textbf{88.94} \\ \hline 
\multicolumn{1}{c|}{}                              & NDCG & 30.36\textsubscript{\textcolor{purple}{\scriptsize\textbf{$\downarrow$ 9.58}}} & \textbf{39.94} \\ \cline{2-4} 
\multicolumn{1}{c|}{\multirow{-2}{*}{Foursquare}}  & HR   & 45.06\textsubscript{\textcolor{purple}{\scriptsize\textbf{$\downarrow$ 8.83}}} & \textbf{53.89} \\ \hline 
\multicolumn{1}{c|}{}                              & NDCG & 15.60\textsubscript{\textcolor{purple}{\scriptsize\textbf{$\downarrow$ 0.76}}} & \textbf{16.36} \\ \cline{2-4} 
\multicolumn{1}{c|}{\multirow{-2}{*}{Lastfm}}      & HR   & 23.20\textsubscript{\textcolor{purple}{\scriptsize\textbf{$\downarrow$ 0.80}}} & \textbf{24.00} \\ \hline
\multicolumn{1}{c|}{}                              & NDCG & 71.80\textsubscript{\textcolor{purple}{\scriptsize\textbf{$\downarrow$ 7.90}}} & \textbf{79.70} \\ \cline{2-4} 
\multicolumn{1}{c|}{\multirow{-2}{*}{KGRec-music}} & HR   & 92.34\textsubscript{\textcolor{purple}{\scriptsize\textbf{$\downarrow$ 3.43}}} & \textbf{95.77} \\ \hline
\end{tabular}
\end{table}
To analyze our learned causality, we compared explanations generated by CausalRec, SASRec, and NARM (Fig. \ref{fig:enter-label}). Each row shows a user’s history of actions (left to middle). The colored box on the left represents the explanation for the model recommendation, and the box on the right is the ground truth (orange) alongside models' recommendations: CausalRec (green), NARM (red), and SASRec (purple). We bold two elements: the film name if it is in the same series as the target, and the category if it matches that of the target film. We also show the learned causal graph of \ours{} on the right. In four examples, CausalRec consistently identifies another film from the same series as a causal explanation, which is intuitively sound. SASRec emphasizes category similarity, capturing thematic but less direct links, while NARM often provides unclear explanations with weak ties. These comparisons show that CausalRec more precisely identifies the causal link between user history and the recommendation. NARM or SASRec often locate semantically relevant items but miss the actual cause-and-effect relationship. These findings align with our experimental results, emphasizing that incorporating a causal relationship can substantially enhance the clarity and correctness of recommendation explanations.

\section{Conclusion}
We proposed \ours{}, a CausalBoost sequential recommendation model that integrates causality into attention mechanisms. To our knowledge, \ours{} is the first sequential recommendation model that integrates causal attention, effectively learns item-level causal relationships, and integrates them into the attention mechanism, enhancing predictions by emphasizing items with genuine causal effects on user preferences rather than mere correlations. In the transformer framework, \ours{} follows the causal identifiability condition as Lemma~\ref{lm1} due to the attention mechanism with layer normalization and achieves the best effect with a causal booster process. Experiments on four real-world datasets demonstrated \ours{}'s superiority over state-of-the-art baselines in NDCG and HR metrics. The visual results show that \ours{} can capture underlying user patterns, such as a preference for film series, underscoring the interpretability benefits of modeling causality. 


%% file: tab_abla.tex
\begin{table*}
\caption{Comparison of CausalRec and its variants in terms of NDCG@10 and HR@10. All the numbers are percentage values with "\%" omitted.}
\label{tab:abla}
\centering
\begin{tabular}{l|cc|cc|cc|cc}
\hline
Dataset  & \multicolumn{2}{c|}{Movielen-1m}       & \multicolumn{2}{c|}{Foursquare} & \multicolumn{2}{c|}{Lastfm} & \multicolumn{2}{c}{KGRec-music}      \\ \hline
{{Metric@10}} & NDCG  & HR    & NDCG  & HR    & NDCG  & HR    & NDCG  & HR    \\ \hline
CausalRec(w/o   Attention)             & 62.09\textsubscript{\textcolor{cyan}{\scriptsize\textbf{$\downarrow$ 10.31}}} & 85.38\textsubscript{\textcolor{cyan}{\scriptsize\textbf{$\downarrow$ 3.56}}} & 36.34\textsubscript{\textcolor{cyan}{\scriptsize\textbf{$\downarrow$ 3.60}}} & 48.10\textsubscript{\textcolor{cyan}{\scriptsize\textbf{$\downarrow$ 5.79}}}  &  16.15\textsubscript{\textcolor{cyan}{\scriptsize\textbf{$\downarrow$ 0.21}}} & 23.26\textsubscript{\textcolor{cyan}{\scriptsize\textbf{$\downarrow$ 0.74}}} &
64.46\textsubscript{\textcolor{cyan}{\scriptsize\textbf{$\downarrow$ 15.24}}} & 91.47\textsubscript{\textcolor{cyan}{\scriptsize\textbf{$\downarrow$ 4.30}}} \\ \hline
CausalRec(w/o   Causality)             & 59.97\textsubscript{\textcolor{red}{\scriptsize\textbf{$\downarrow$ 12.43}}} & 82.92\textsubscript{\textcolor{red}{\scriptsize\textbf{$\downarrow$ 6.02}}} & 20.02\textsubscript{\textcolor{red}{\scriptsize\textbf{$\downarrow$ 19.92}}} & 35.92\textsubscript{\textcolor{red}{\scriptsize\textbf{$\downarrow$ 17.97}}} &  13.43\textsubscript{\textcolor{red}{\scriptsize\textbf{$\downarrow$ 2.93}}} & 18.30\textsubscript{\textcolor{red}{\scriptsize\textbf{$\downarrow$ 5.70}}}  &
75.32\textsubscript{\textcolor{red}{\scriptsize\textbf{$\downarrow$ 4.38}}} & 93.15\textsubscript{\textcolor{red}{\scriptsize\textbf{$\downarrow$ 2.62}}} \\ \hline

CausalRec(w/o   sparse)               & 62.96\textsubscript{\textcolor{blue}{\scriptsize\textbf{$\downarrow$ 9.44}}} & 86.47\textsubscript{\textcolor{blue}{\scriptsize\textbf{$\downarrow$ 2.47}}} & 25.04\textsubscript{\textcolor{blue}{\scriptsize\textbf{$\downarrow$ 14.90}}} & 38.96\textsubscript{\textcolor{blue}{\scriptsize\textbf{$\downarrow$ 14.93}}} &  15.28\textsubscript{\textcolor{blue}{\scriptsize\textbf{$\downarrow$ 1.08}}} & 21.61\textsubscript{\textcolor{blue}{\scriptsize\textbf{$\downarrow$ 2.39}}} &
75.81\textsubscript{\textcolor{blue}{\scriptsize\textbf{$\downarrow$ 3.89}}} & 93.38\textsubscript{\textcolor{blue}{\scriptsize\textbf{$\downarrow$ 2.39}}} \\ \hline
CausalRec & \textbf{72.40} & \textbf{88.94} & \textbf{39.94} & \textbf{53.89} & \textbf{16.36} & \textbf{24.00} & \textbf{79.70} & \textbf{95.77} \\ \hline
\end{tabular}
\end{table*}

%% file: GenAI.tex
\section{Ethical Considerations}
\label{sec:eth}
CausalRec, a sequential recommendation model integrating causal attention, presents potential ethical risks, including biased recommendations due to implicit biases in historical user behavior data (which may disadvantage specific user groups), and security risks of malicious attacks on causal graph data. Mitigation strategies involve bias auditing of training data using fairness-aware machine learning techniques, applying differential privacy to user interaction sequences and encrypted storage of causal graphs, conducting regular penetration testing on the model architecture, and establishing transparent audit trails for recommendation logic to ensure alignment with ethical principles and user welfare.